\documentclass[11pt]{article}
 
\usepackage{amsmath,amssymb,amsthm}
\usepackage{authblk}
\usepackage{subfig}
\usepackage[T1]{fontenc}
\usepackage{tikz}
\usetikzlibrary{arrows,positioning,automata,calc,shapes,intersections}
\usepackage[linesnumbered,ruled,vlined]{algorithm2e}
\usepackage{bm}
\usepackage{CJKutf8}

\newcommand{\LPk}{\mathrm{LP}_k}
\newcommand{\LP}{\mathrm{LP}}
\newcommand{\ot}{\leftarrow}
\newcommand{\argmax}{\mathop{\rm argmax}}
\newcommand{\argmin}{\mathop{\rm argmin}}

\theoremstyle{plain}
\newtheorem{theorem}     {Theorem}
\newtheorem{lemma}       {Lemma}

\newtheorem{corollary}   {Corollary}

\newtheorem{fact}        {Fact}

\makeatletter
\newcommand{\subalign}[1]{%
  \vcenter{%
    \Let@ \restore@math@cr \default@tag
    \baselineskip\fontdimen10 \scriptfont\tw@
    \advance\baselineskip\fontdimen12 \scriptfont\tw@
    \lineskip\thr@@\fontdimen8 \scriptfont\thr@@
    \lineskiplimit\lineskip
    \ialign{\hfil$\m@th\scriptstyle##$&$\m@th\scriptstyle{}##$\crcr
      #1\crcr
    }%
  }
}
\makeatother

\usepackage[top=30truemm,bottom=30truemm,left=25truemm,right=25truemm]{geometry}

\usepackage{setspace}
\title{{\Large The Densest Subgraph Problem with a Convex/Concave Size Function}\footnote{A preliminary version of this paper appeared in the \emph{Proceedings of the 27th International Symposium on Algorithms and Computation (ISAAC~2016)} \cite{Kawase_Miyauchi_16}.}}

\author{Yasushi Kawase}
\author{Atsushi Miyauchi}
\affil{Tokyo Institute of Technology, Tokyo, Japan\\
  \texttt{\{kawase.y.ab,~miyauchi.a.aa\}@m.titech.ac.jp}}
\date{\empty}

\begin{document}
\maketitle

\begin{abstract}
In the densest subgraph problem, given an edge-weighted undirected graph $G=(V,E,w)$, 
we are asked to find $S\subseteq V$ that maximizes the \emph{density}, 
i.e., $w(S)/|S|$, where $w(S)$ is the sum of weights of the edges in the subgraph induced by $S$. 
This problem has often been employed in a wide variety of graph mining applications. 
However, the problem has a drawback; 
it may happen that the obtained subset is too large or too small in comparison with the size desired in the application at hand. 
In this study, we address the size issue of the densest subgraph problem by generalizing the density of $S\subseteq V$.  
Specifically, we introduce the \emph{$f$-density} of $S\subseteq V$, which is defined as $w(S)/f(|S|)$, 
where $f:\mathbb{Z}_{\geq 0}\rightarrow \mathbb{R}_{\geq 0}$ is a monotonically non-decreasing function. 
In the \emph{$f$-densest subgraph problem} ($f$-DS), 
we aim to find $S\subseteq V$ that maximizes the $f$-density $w(S)/f(|S|)$. 
Although $f$-DS does not explicitly specify the size of the output subset of vertices, 
we can handle the above size issue 
using a \emph{convex}/\emph{concave} size function $f$ appropriately. 
For $f$-DS with convex function $f$, we propose a nearly-linear-time algorithm with a provable approximation guarantee.  
On the other hand, for $f$-DS with concave function $f$, 
we propose an LP-based exact algorithm, a flow-based $O(|V|^3)$-time exact algorithm for unweighted graphs, 
and a nearly-linear-time approximation algorithm. 
\end{abstract}

\section{Introduction}
Finding dense components in a graph is an active research topic in graph mining. 
Techniques for identifying dense subgraphs have been used in various applications. 
For example, in Web graph analysis, they are used for detecting communities 
(i.e., sets of web pages dealing with the same or similar topics)~\cite{Dourisboure+_07} 
and spam link farms~\cite{Gibson+_05}. 
As another example, in bioinformatics, 
they are used for finding molecular complexes in protein--protein interaction networks~\cite{Bader_Hogue_03} 
and identifying regulatory motifs in DNA~\cite{Fratkin+_06}.  
Furthermore, they are also used for expert team formation~\cite{Bonchi+_14,Tsourakakis+_13} 
and real-time story identification in micro-blogging streams~\cite{Angel+_12}. 

To date, various optimization problems have been considered to find dense components in a graph. 
The densest subgraph problem is one of the most well-studied optimization problems.  
Let $G=(V,E,w)$ be an edge-weighted undirected graph consisting of $n=|V|$ vertices, $m=|E|$ edges, 
and a weight function $w:E\rightarrow \mathbb{Q}_{>0}$, where $\mathbb{Q}_{>0}$ is the set of positive rational numbers. 
For a subset of vertices $S\subseteq V$, let $G[S]$ be the subgraph induced by $S$, 
i.e., $G[S]=(S,E(S))$, where $E(S)=\{\{i,j\}\in E\mid i,j\in S\}$. 
The \emph{density} of $S\subseteq V$ is defined as $w(S)/|S|$, where $w(S)=\sum_{e\in E(S)}w(e)$. 
In the (weighted) densest subgraph problem, given an (edge-weighted) undirected graph $G=(V,E,w)$, 
we are asked to find $S\subseteq V$ that maximizes the density $w(S)/|S|$. 

The densest subgraph problem has received significant attention recently 
because it can be solved exactly in polynomial time and approximately in nearly linear time.  
In fact, there exist a flow-based exact algorithm~\cite{Goldberg_84} 
and a linear-programming-based (LP-based) exact algorithm~\cite{Charikar_00}. 
Charikar~\cite{Charikar_00} demonstrated that 
the greedy algorithm designed by Asahiro et al.~\cite{Asahiro+_00}, which is called the \emph{greedy peeling}, 
obtains a $2$-approximate solution\footnote{
A feasible solution is said to be \emph{$\gamma$-approximate} 
if its objective value times $\gamma$ is greater than or equal to the optimal value. 
An algorithm is called a \emph{$\gamma$-approximation algorithm} 
if it runs in polynomial time and returns a $\gamma$-approximate solution for any instance. 
For a $\gamma$-approximation algorithm, $\gamma$ is referred to as an \emph{approximation ratio} of the algorithm. 
} for any instance. 
This algorithm runs in $O(m+n\log n)$ time for weighted graphs and $O(m+n)$ time for unweighted graphs. 

However, the densest subgraph problem has a drawback;
it may happen that the obtained subset is too large or too small in comparison with the size desired in the application at hand. 
To overcome this issue, some variants of the problem have often been employed. 
The densest $k$-subgraph problem (D$k$S) is a straightforward size-restricted variant of the densest subgraph problem~\cite{FeKoPe01}.  
In this problem, given an additional input $k$ being a positive integer, 
we are asked to find $S\subseteq V$ of size $k$ that maximizes the density $w(S)/|S|$. 
Note that in this problem, the objective function can be replaced by $w(S)$ since $|S|$ is fixed to $k$. 
Unfortunately, it is known that this size restriction makes the problem much harder to solve. 
In fact, Khot~\cite{Khot_06} proved that D$k$S has no PTAS under some reasonable computational complexity assumption. 
The current best approximation algorithm has an approximation ratio of $O(n^{1/4+\epsilon})$ 
for any $\epsilon >0$~\cite{Bhaskara+_10}. 

Furthermore, Andersen and Chellapilla~\cite{Andersen_Chellapilla_09} introduced two relaxed versions of D$k$S. 
The first problem, the densest at-least-$k$-subgraph problem (Dal$k$S), 
asks for $S\subseteq V$ that maximizes the density $w(S)/|S|$ under the size constraint $|S|\geq k$. 
For this problem, Andersen and Chellapilla~\cite{Andersen_Chellapilla_09} adopted the greedy peeling, 
and demonstrated that the algorithm yields a $3$-approximate solution for any instance. 
Later, Khuller and Saha~\cite{Khuller_Saha_09} investigated the problem more deeply. 
They proved that Dal$k$S is NP-hard, and designed a flow-based algorithm and an LP-based algorithm. 
These algorithms have an approximation ratio of $2$, which improves the above approximation ratio of $3$. 
The second problem is called the densest at-most-$k$-subgraph problem (Dam$k$S), 
which asks for $S\subseteq V$ that maximizes the density $w(S)/|S|$ under the size constraint $|S|\leq k$. 
The NP-hardness is immediate since finding a maximum clique can be reduced to it. 
Khuller and Saha~\cite{Khuller_Saha_09} proved that 
approximating Dam$k$S is as hard as approximating D$k$S, within a constant factor.

\subsection{Our Contribution}
In this study, we address the size issue of the densest subgraph problem by generalizing the density of $S\subseteq V$. 
Specifically, we introduce the \emph{$f$-density} of $S\subseteq V$, which is defined as $w(S)/f(|S|)$, 
where $f: \mathbb{Z}_{\geq 0}\rightarrow \mathbb{R}_{\geq 0}$ 
is a monotonically non-decreasing function with $f(0)=0$.\footnote{
To handle various types of functions (e.g., $f(x)=x^\alpha$ for $\alpha >0$), 
we set the codomain of the function $f$ to be the set of nonnegative \emph{real} numbers. 
We assume that we can compare $p\cdot f(i)$ and $q\cdot f(j)$ in constant time 
for any $p,q\in\mathbb{Q}$ and $i,j\in\mathbb{Z}_{\geq 0}$.
} 
Note that $\mathbb{Z}_{\geq 0}$ and $\mathbb{R}_{\geq 0}$ are 
the sets of nonnegative integers and nonnegative real numbers, respectively. 
In the \emph{$f$-densest subgraph problem} ($f$-DS), 
we aim to find $S\subseteq V$ that maximizes the $f$-density $w(S)/f(|S|)$.
For simplicity, we assume that $E\neq \emptyset$. Hence, any optimal solution $S^*\subseteq V$ satisfies $|S^*|\geq 2$. 
Although $f$-DS does not explicitly specify the size of the output subset of vertices, 
we can handle the above size issue 
using a \emph{convex} size function $f$ or a \emph{concave} size function $f$ appropriately. 
In fact, we can show that any optimal solution to $f$-DS with convex (resp. concave) function $f$ has a size 
smaller (resp. larger) than or equal to that of any densest subgraph (i.e., any optimal solution to the densest subgraph problem). 
For details, see Sections~\ref{sec:convex} and \ref{sec:concave}. 

Here we mention the relationship between our problem and D$k$S. 
Any optimal solution $S^*\subseteq V$ to $f$-DS is a maximum weight subset of size $|S^*|$, 
i.e., $S^*\in \argmax\{w(S)\mid S\subseteq V,~|S|=|S^*|\}$, 
which implies that $S^*$ is also optimal to D$k$S with $k=|S^*|$. 
Furthermore, the iterative use of a $\gamma$-approximation algorithm for D$k$S leads to a $\gamma$-approximation algorithm for $f$-DS. 
Using the above $O(n^{1/4+\epsilon})$-approximation algorithm for D$k$S~\cite{Bhaskara+_10}, 
we can obtain an $O(n^{1/4+\epsilon})$-approximation algorithm for $f$-DS. 

In what follows, we summarize our results for both the cases where $f$ is convex and where $f$ is concave.

\paragraph*{The case where $f$ is convex.} 
We first describe our results for the case where $f$ is convex. 
A function $f:\mathbb{Z}_{\geq 0}\rightarrow \mathbb{R}_{\geq 0}$ is said to be \emph{convex} 
if $f(x)-2f(x+1)+f(x+2)\ge 0$ holds for any $x\in\mathbb{Z}_{\ge 0}$.
We first prove the NP-hardness of $f$-DS with a certain convex function $f$ 
by constructing a reduction from Dam$k$S.
Thus, for $f$-DS with convex function $f$, 
one of the best possible ways is to design an algorithm with a provable approximation guarantee. 

To this end, we propose a 
$\min\left\{\frac{f(2)/2}{f(|S^*|)/|S^*|^2},~\frac{2f(n)/n}{f(|S^*|)-f(|S^*|-1)}\right\}$-approximation algorithm, 
where $S^*\subseteq V$ is an optimal solution to $f$-DS with convex function $f$. 
Our algorithm consists of the following two procedures, and outputs the better solution found by them. 
The first one is based on the brute-force search, 
which obtains an $\frac{f(2)/2}{f(|S^*|)/|S^*|^2}$-approximate solution in $O(m+n)$ time. 
The second one adopts the greedy peeling, 
which obtains a $\frac{2f(n)/n}{f(|S^*|)-f(|S^*|-1)}$-approximate solution in $O(m+n\log n)$ time. 
Thus, the total running time of our algorithm is $O(m+n\log n)$. 
Our analysis on the approximation ratio of the second procedure extends 
the analysis by Charikar~\cite{Charikar_00} for the densest subgraph problem. 

At the end of our analysis, 
we observe the behavior of the approximation ratio of our algorithm for three concrete size functions. 
We consider size functions \emph{between} linear and quadratic
because, as we will see later, 
$f$-DS with any super-quadratic size function is a trivial problem; in fact, it only produces constant-size optimal solutions.
The first example is $f(x)=x^\alpha~(\alpha \in [1, 2])$. 
We show that the approximation ratio of our algorithm is $2\cdot n^{(\alpha-1)(2-\alpha)}$, 
where the worst-case performance of $2\cdot n^{1/4}$ is attained at $\alpha=1.5$. 
The second example is $f(x)=\lambda x+(1-\lambda)x^2~(\lambda \in [0,1))$. 
For this case, the approximation ratio of our algorithm is $(2-\lambda)/(1-\lambda)$, 
which is a constant for a fixed $\lambda$.
The third example is $f(x)=x^2/(\lambda x +(1-\lambda))~(\lambda \in [0,1])$. 
Note that this size function is derived by density function $\lambda \frac{w(S)}{|S|}+(1-\lambda)\frac{w(S)}{|S|^2}$. 
The approximation ratio of our algorithm is $4/(1+\lambda)$, which is at most $4$.

\paragraph*{The case where $f$ is concave.} 
We next describe our results for the case where $f$ is concave. 
A function $f:\mathbb{Z}_{\geq 0}\rightarrow \mathbb{R}_{\geq 0}$ is said to be \emph{concave} 
if $f(x)-2f(x+1)+f(x+2)\le 0$ holds for any $x\in\mathbb{Z}_{\ge 0}$.
Unlike the above convex case, $f$-DS in this case can be solved exactly in polynomial time.

In fact, we present an LP-based exact algorithm, 
which extends Charikar's exact algorithm for the densest subgraph problem~\cite{Charikar_00} 
and Khuller and Saha's $2$-approximation algorithm for Dal$k$S~\cite{Khuller_Saha_09}. 
It should be emphasized that our LP-based algorithm obtains 
not only an optimal solution to $f$-DS but also some attractive subsets of vertices. 
Let us see an example in Figure~\ref{fig:dense_frontier}. 
The graph consists of 8 vertices and 11 unweighted edges (i.e., $w(e)=1$ for every $e\in E$). 
For this graph, we plotted all the points contained in $\mathcal{P}=\{(|S|,w(S))\mid S\subseteq V\}$. 
We refer to the extreme points of the upper convex hull of $\mathcal{P}$ as the \emph{dense frontier points}. 
The (smallest) densest subgraph is a typical subset of vertices corresponding to a dense frontier point.
Our LP-based algorithm obtains a corresponding subset of vertices for every dense frontier point. 
It should be noted that the algorithm \textsf{SSM} designed by Nagano, Kawahara, and Aihara~\cite{NaKaAi11} can also be used to 
obtain a corresponding subset of vertices for every dense frontier point. 
The difference between their algorithm and ours is that 
their algorithm is based on the computation of a minimum norm base, 
whereas ours solves linear programming problems. 
\begin{figure}[t]
\centering
\includegraphics[scale=1.12]{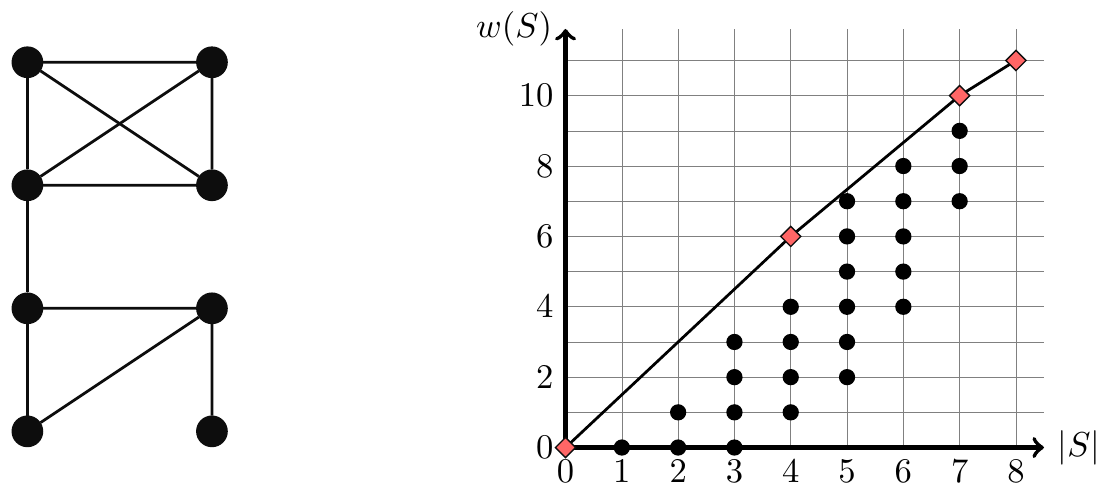}
\caption{An example graph and the corresponding points in $\mathcal{P}=\{(|S|,w(S))\mid S\subseteq V\}$. 
The diamond-shaped points, 
i.e., $(0,0),(4,6),(7,10)$, and $(8,11)$,
are dense frontier points.}\label{fig:dense_frontier}
\end{figure}

Moreover, in this concave case, we design a combinatorial exact algorithm for unweighted graphs.
Our algorithm is based on the standard technique for fractional programming.
By using the technique, we can reduce $f$-DS to a sequence of submodular function minimizations. 
However, the direct application of a submodular function minimization algorithm leads to a computationally expensive algorithm 
that runs in $O(n^5(m+n)\cdot \log n)$ time.
To reduce the computation time, 
we replace a submodular function minimization algorithm with a much faster flow-based algorithm  
that substantially extends a technique of Goldberg's flow-based algorithm 
for the densest subgraph problem~\cite{Goldberg_84}. 
The total running time of our algorithm is $O(n^3)$.
Modifying this algorithm, we also present an $O\left(\frac{n^3}{\log n}\cdot \log\left(\frac{\log n}{\epsilon}\right)\right)$-time 
$(1+\epsilon)$-approximation algorithm for weighted graphs. 

Although our flow-based algorithm is much faster than the reduction-based algorithm, 
the running time is still long for large-sized graphs. 
To design an algorithm with much higher scalability, 
we adopt the greedy peeling.  
As mentioned above, this algorithm runs in 
$O(m+n\log n)$ time for weighted graphs and $O(m+n)$ time for unweighted graphs.
We prove that the algorithm yields a $3$-approximate solution for any instance.

\subsection{Related Work}
Tsourakakis et al.~\cite{Tsourakakis+_13} introduced a general optimization problem to find dense subgraphs, 
which is referred to as the optimal $(g,h,\alpha)$-edge-surplus problem. 
In this problem, given an unweighted undirected graph $G=(V,E)$, we are asked to find $S\subseteq V$ that maximizes 
\(\textsf{edge-surplus}_\alpha(S)=g(|E(S)|) -\alpha h(|S|)\), 
where $g$ and $h$ are strictly monotonically increasing functions, and $\alpha >0$ is a constant. 
The intuition behind this optimization problem is the same as that of $f$-DS. 
In fact, the first term $g(|E(S)|)$ prefers $S\subseteq V$ that has a large number of edges, 
whereas the second term $-\alpha h(|S|)$ penalizes $S\subseteq V$ with a large size.  
Tsourakakis et al.~\cite{Tsourakakis+_13} were motivated by finding near-cliques 
(i.e., relatively small dense subgraphs), 
and they derived the function $\textsf{OQC}_\alpha(S) =|E(S)|-\alpha \binom{|S|}{2}$, 
which is called the OQC function, by setting $g(x)=x$ and $h(x)=x(x-1)/2$. 
For OQC function maximization, they adopted the greedy peeling and a simple local search heuristic. 

Recently, Yanagisawa and Hara~\cite{Yanagisawa_Hara_16} introduced 
density function $|E(S)|/|S|^\alpha$ for $\alpha \in (1,2]$, 
which they called the discounted average degree. 
For discounted average degree maximization, 
they designed an integer-programming-based exact algorithm, 
which is applicable only to graphs with a maximum of a few thousand edges. 
They also designed a local search heuristic, 
which is applicable to web-scale graphs but has no provable approximation guarantee. 
As mentioned above, our algorithm for $f$-DS with convex function $f$ runs in $O(m+n\log n)$ time, 
and has an approximation ratio of $2\cdot n^{(\alpha-1)(2-\alpha)}$ for $f(x)=x^\alpha$ ($\alpha \in [1,2]$).

\section{Convex Case}\label{sec:convex}
In this section, we investigate $f$-DS with convex function $f$. 
A function $f:\mathbb{Z}_{\geq 0}\rightarrow \mathbb{R}_{\geq 0}$ is said to be \emph{convex} 
if $f(x)-2f(x+1)+f(x+2)\ge 0$ holds for any $x\in\mathbb{Z}_{\ge 0}$. 
We remark that $f(x)/x$ is monotonically non-decreasing for $x$ since we assume that $f(0)=0$. 
It should be emphasized that any optimal solution to $f$-DS with convex function $f$ 
has a size smaller than or equal to that of any densest subgraph. 
To see this, let $S^*\subseteq V$ be any optimal solution to $f$-DS 
and $S^*_\text{DS}\subseteq V$ be any densest subgraph.
Then we have
\begin{align}
  \frac{f(|S^*|)}{|S^*|}
  =\frac{w(S^*)/|S^*|}{w(S^*)/f(|S^*|)}
  \le \frac{w(S^*_\text{DS})/|S^*_\text{DS}|}{w(S^*_\text{DS})/f(|S^*_\text{DS}|)}
  =\frac{f(|S^*_\text{DS}|)}{|S^*_\text{DS}|}.\label{eq:size}
\end{align}
This implies that $|S^*|\leq |S^*_\text{DS}|$ holds because $f(x)/x$ is monotonically non-decreasing. 

\subsection{Hardness}
We first prove that $f$-DS with convex function $f$ contains Dam$k$S as a special case.
\begin{theorem}\label{thm:hardness}
For any integer $k\in [2,n]$,
$S\subseteq V$ is optimal to Dam$k$S if and only if
$S$ is optimal to $f$-DS with (convex) function 
$f(x)=\max\left\{x,~\frac{w(V)}{w(e)/2}(x-k)+k\right\}$, where $e$ is an arbitrary edge. 
\end{theorem}
\begin{proof}
Since the maximum of linear functions is convex, the function $f$ is convex. 
We remark that 
\begin{align*}
f(x)=
\begin{cases}
	x &\text{if } x\leq k,\\
	\frac{w(V)}{w(e)/2}(x-k)+k & \text{otherwise}. 
\end{cases}
\end{align*}
For any $S\subseteq V$ with $|S|\le k$, we have $w(S)/f(|S|)=w(S)/|S|$. 
On the other hand, for any $S\subseteq V$ with $|S|>k$, we have
\begin{align*}
	\frac{w(S)}{f(|S|)}
	= \frac{w(S)}{\frac{w(V)}{w(e)/2}(|S|-k)+k}
	<\frac{w(S)}{\frac{w(V)}{w(e)/2}}
	\leq \frac{w(e)}{2}, 
\end{align*}
which implies that $S$ is not optimal to $f$-DS. 
Thus, we have the theorem.
\end{proof}

\subsection{Our Algorithm}
In this subsection, we provide an algorithm for $f$-DS with convex function $f$. 
Our algorithm consists of the following two procedures, and outputs the better solution found by them. 
Let $S^*\subseteq V$ be an optimal solution to the problem. 
The first one is based on the brute-force search, 
which obtains an $\frac{f(2)/2}{f(|S^*|)/|S^*|^2}$-approximate solution in $O(m+n)$ time. 
The second one adopts the greedy peeling~\cite{Asahiro+_00}, 
which obtains a $\frac{2f(n)/n}{f(|S^*|)-f(|S^*|-1)}$-approximate solution in $O(m+n\log n)$ time. 
Combining these results, both of which will be proved later, we have the following theorem. 
\begin{theorem}\label{thm:convex}
Let $S^*\subseteq V$ be an optimal solution to $f$-DS with convex function $f$. 
For the problem, our algorithm runs in $O(m+n\log n)$ time, and has an approximation ratio of 
\begin{align*}
\min\left\{\frac{f(2)/2}{f(|S^*|)/|S^*|^2},~\frac{2f(n)/n}{f(|S^*|)-f(|S^*|-1)}\right\}.
\end{align*}
\end{theorem}

\subsubsection{Brute-Force Search}\label{subsec:small}
As will be shown below, to obtain an $\frac{f(2)/2}{f(|S^*|)/|S^*|^2}$-approximate solution, 
it suffices to find the heaviest edge (i.e., $\argmax\{w(e)\mid e\in E\}$), 
which can be done in $O(m+n)$ time. 
Here we present a more general algorithm, which is useful for some case. 
Our algorithm examines all the subsets of vertices of size at most $k$, 
and then returns an optimal subset among them, where $k$ is a constant that satisfies $k\ge 2$. 
For reference, we describe the procedure in Algorithm~\ref{alg:brute}. 
This algorithm can be implemented to run in $O((m+n)n^k)$ time
because the number of subsets with at most $k$ vertices is $\sum_{i=0}^k\binom{n}{i}=O(n^k)$ and
the value of $w(S)/f(|S|)$ for each $S\subseteq V$ can be computed in $O(m+n)$ time. 
\begin{algorithm}[t] 
\caption{Brute-force search}\label{alg:brute}
\For{$i\ot 2,\dots, k$}{
  Find $S^*_i\in \argmax\{w(S)\mid S\subseteq V,~|S|=i\}$ by examining all the candidate subsets\;
}
\Return $S\in \{S^*_2,\dots,S^*_k\}$ that maximizes $w(S)/f(|S|)$\;
\end{algorithm}

We analyze the approximation ratio of the algorithm. 
Let $S_i^*\subseteq V$ denote a maximum weight subset of size $i\geq 2$, i.e., $S_i^*\in \argmax\{w(S)\mid S\subseteq V,~|S|=i\}$.
We refer to $w(S_i^*)/\binom{i}{2}$ as the \emph{edge density} of $i$ vertices. 
The following lemma gives a fundamental property of the edge density. 
\begin{lemma}\label{lem:edgedensity}
The edge density is monotonically non-increasing for the number of vertices, 
i.e., \(w(S_i^*)/\binom{i}{2}\ge w(S_j^*)/\binom{j}{2}\) holds for any $2\le i\le j\le n$.
\end{lemma}
\begin{proof}
It suffices to show that 
$w(S^*_i)/\binom{i}{2}\ge w(S^*_{i+1})/\binom{i+1}{2}$ holds for any positive integer $i\in [2,n-1]$. 
For $S\subseteq V$ and $v\in S$, 
let $d_S(v)$ denote the weighted degree of $v$ in the induced subgraph $G[S]$, 
i.e., $d_S(v)=\sum_{u\in V:\,\{u,v\}\in E(S)}w(\{u,v\})$. 
Take a vertex $u\in \argmin\{d_{S^*_{i+1}}(v)\mid v\in S^*_{i+1}\}$.
Then we obtain \(d_{S^*_{i+1}}(u)\le \frac{1}{i+1}\sum_{v\in S^*_{i+1}}d_{S^*_{i+1}}(v)= \frac{2}{i+1}\cdot w(S^*_{i+1})\).
Hence, we have
\begin{align*}
\frac{w(S^*_i)}{\binom{i}{2}}
\ge \frac{w(S^*_{i+1}\setminus \{u\})}{\binom{i}{2}}
=\frac{w(S^*_{i+1})-d_{S^*_{i+1}}(u)}{\binom{i}{2}}
\ge \frac{(1-\frac{2}{i+1})\cdot w(S^*_{i+1})}{\binom{i}{2}}
=\frac{w(S^*_{i+1})}{\binom{i+1}{2}}, 
\end{align*}
as desired.
\end{proof}

Using the above lemma, we can provide the approximation ratio. 
\begin{lemma}\label{lem:convex_small}
Let $S^*\subseteq V$ be an optimal solution to $f$-DS with convex function $f$. 
If $|S^*|\leq k$, then Algorithm~\ref{alg:brute} obtains an optimal solution. 
If $|S^*|\geq k$, then it holds that
  \begin{align*}
    \frac{w(S^*)}{f(|S^*|)}\le
    \frac{2\cdot f(k)/k^2}{f(|S^*|)/|S^*|^2}\cdot \frac{w(S^*_k)}{f(k)}.
  \end{align*}
\end{lemma}
\begin{proof}
If $|S^*|\le k$, then Algorithm~\ref{alg:brute} obtains an optimal solution because $S^*\in \{S^*_2,\dots,S^*_k\}$.
If $|S^*|\ge k$, then we have 
  \begin{align*}
    \frac{w(S^*)}{f(|S^*|)}
    &\leq \frac{w(S^*)}{f(|S^*|)}\cdot \frac{w(S_k^*)/\binom{k}{2}}{w(S^*)/\binom{|S^*|}{2}}
	= \frac{f(k)/\binom{k}{2}}{f(|S^*|)/\binom{|S^*|}{2}}\cdot\frac{w(S^*_k)}{f(k)}\\
	&= \frac{1-1/|S^*|}{1-1/k}\cdot \frac{f(k)/k^2}{f(|S^*|)/|S^*|^2}\cdot\frac{w(S^*_k)}{f(k)}
    \leq \frac{2\cdot f(k)/k^2}{f(|S^*|)/|S^*|^2}\cdot\frac{w(S^*_k)}{f(k)},
  \end{align*}
  where the first inequality follows from Lemma~\ref{lem:edgedensity}, 
  and the last inequality follows from $k\ge 2$.
\end{proof}
From this lemma, we see that Algorithm~\ref{alg:brute} with $k=2$ 
has an approximation ratio of $\frac{f(2)/2}{f(|S^*|)/|S^*|^2}$.

\subsubsection{Greedy Peeling}\label{subsec:large}
Here we adopt the greedy peeling.  
For $S\subseteq V$ and $v\in S$, 
let $d_S(v)$ denote the weighted degree of $v$ in the induced subgraph $G[S]$, 
i.e., $d_S(v)=\sum_{u\in V:\,\{u,v\}\in E(S)}w(\{u,v\})$. 
Our algorithm iteratively removes the vertex with the smallest weighted degree in the currently remaining graph, 
and then returns $S\subseteq V$ with maximum $w(S)/f(|S|)$ over the iterations.  
For reference, we describe the procedure in Algorithm~\ref{alg:peeling}. 
This algorithm can be implemented to run in $O(m+n\log n)$ time for weighted graphs and $O(m+n)$ time for unweighted graphs. 
\begin{algorithm}[t] 
\caption{Greedy peeling}\label{alg:peeling}
$S_n\ot V$\;
\For{$i\ot n,\dots,2$}{
  Find $v_i\in \argmin_{v\in S_i} d_{S_i}(v)$ and $S_{i-1}\ot S_i\setminus\{v_i\}$\;
}
\Return $S\in \{S_1,\dots, S_n\}$ that maximizes $w(S)/f(|S|)$\;
\end{algorithm}

The following lemma provides the approximation ratio. 
\begin{lemma}\label{lem:convex_peeling}
  Let $S^*\subseteq V$ be an optimal solution to $f$-DS with convex function $f$. 
  Algorithm \ref{alg:peeling} returns $S\subseteq V$ that satisfies 
  \begin{align*}
 \frac{w(S^*)}{f(|S^*|)}\leq \frac{2f(n)/n}{f(|S^*|)-f(|S^*|-1)}\cdot \frac{w(S)}{f(|S|)}. 
  \end{align*}
\end{lemma}
\begin{proof}
Choose an arbitrary vertex $v\in S^*$. By the optimality of $S^*$, we have
\begin{align*}
\frac{w(S^*)}{f(|S^*|)}\ge \frac{w(S^*\setminus\{v\})}{f(|S^*|-1)}. 
\end{align*}
By using the fact that $w(S^*\setminus \{v\}) = w(S^*)-d_{S^*}(v)$, the above inequality can be transformed to 
\begin{align}\label{ineq:deg}
d_{S^*}(v)\ge (f(|S^*|)-f(|S^*|-1))\cdot \frac{w(S^*)}{f(|S^*|)}.
\end{align}
Let $l$ be the smallest index that satisfies $S_l\supseteq S^*$, 
where $S_l$ is the subset of vertices of size $l$ appeared in Algorithm~\ref{alg:peeling}. 
Note that $v_l\ (\in \argmin_{v\in S_l} d_{S_l}(v))$ is contained in $S^*$. 
Then we have 
\begin{align*}
  \frac{w(S_l)}{f(l)}
  &= \frac{\sum_{u\in S_l}d_{S_l}(u)}{2f(l)}
  \ge \frac{l\cdot d_{S_l}(v_l)}{2f(l)}
  \ge \frac{d_{S^*}(v_l)}{2f(l)/l}\\
  &\ge \frac{f(|S^*|)-f(|S^*|-1)}{2f(l)/l}\cdot\frac{w(S^*)}{f(|S^*|)}
  \ge \frac{f(|S^*|)-f(|S^*|-1)}{2f(n)/n}\cdot\frac{w(S^*)}{f(|S^*|)}, 
\end{align*}
where the first inequality follows from the greedy choice of $v_l$, 
the second inequality follows from $S_l\supseteq S^*$, 
the third inequality follows from inequality (\ref{ineq:deg}), 
and the last inequality follows from the monotonicity of $f(x)/x$. 
Since Algorithm~\ref{alg:peeling} considers $S_l$ as a candidate subset of the output, we have the lemma. 
\end{proof}

\subsection{Examples}
Here we observe the behavior of the approximation ratio of our algorithm for three concrete convex size functions. 
We consider size functions \emph{between} linear and quadratic because 
$f$-DS with any super-quadratic size function is a trivial problem; in fact, it only produces constant-size optimal solutions. 
This follows from the inequality \(\frac{f(2)/2}{f(|S^*|)/|S^*|^2}\ge 1\) 
(i.e., $f(2)/2\geq f(|S^*|)/|S^*|^2$) by Lemma~\ref{lem:convex_small}. 

\paragraph*{(i) $\boldsymbol{f(x)=x^{\alpha}\ (\alpha \in [1,2])}$.}
The following corollary provides an approximation ratio of our algorithm. 
\begin{corollary}
  For $f$-DS with $f(x)=x^{\alpha}~(\alpha \in [1,2])$, 
  our algorithm has an approximation ratio of $2\cdot n^{(\alpha-1)(2-\alpha)}$. 
\end{corollary}
\begin{proof}
  Let $s=|S^*|$.
  By Theorem~\ref{thm:convex}, the approximation ratio is 
  \begin{align*}
    \min\left\{\frac{f(2)/2}{f(s)/s^2},~\frac{2f(n)/n}{f(s)-f(s-1)}\right\}
    &=\min\left\{2^{\alpha-1}\cdot s^{2-\alpha},~\frac{2n^{\alpha-1}}{s^{\alpha}-(s-1)^{\alpha}}\right\}\\
    &\le \min\left\{2\cdot s^{2-\alpha},~\frac{2n^{\alpha-1}}{s^{\alpha-1}}\right\}
    \le 2\cdot n^{(\alpha-1)(2-\alpha)}. 
  \end{align*}
  The first inequality follows from the fact that 
  $s^\alpha -(s-1)^\alpha = s^\alpha -(s-1)^{\alpha -1}(s-1)\geq s^\alpha -s^{\alpha -1}(s-1)=s^{\alpha -1}$. 
  The last inequality follows from the fact that the first term and the second term of the minimum function are, respectively, 
  monotonically non-decreasing and non-increasing for $s$, 
  and they have the same value at $s=n^{\alpha-1}$.
\end{proof}
Note that an upper bound on $2\cdot n^{(\alpha-1)(2-\alpha)}$ is $2\cdot n^{1/4}$, 
which is attained at $\alpha=1.5$.

\paragraph*{(ii) $\boldsymbol{f(x)=\lambda x+(1-\lambda)x^2\ (\lambda \in [0,1))}$.}
The following corollary provides an approximation ratio of Algorithm~\ref{alg:brute}, 
which is a constant for a fixed $\lambda$. 
\begin{corollary}\label{cor:ex2}
  For $f$-DS with $f(x)=\lambda x+(1-\lambda)x^2~(\lambda \in [0,1))$,
  Algorithm~\ref{alg:brute} with $k=2$ has an approximation ratio of $(2-\lambda)/(1-\lambda)$.
  Furthermore, for any $\epsilon >0$, 
  Algorithm~\ref{alg:brute} with $k\ge \frac{2}{\epsilon}\cdot\frac{\lambda}{1-\lambda}$ 
  has an approximation ratio of $2+\epsilon$. 
\end{corollary}
\begin{proof}
  Let $s=|S^*|$.
  By Lemma~\ref{lem:convex_small}, the approximation ratio is 
  \begin{align*}
    \frac{2\cdot f(k)/k^2}{f(s)/s^2}
    =2\cdot \frac{\lambda/k+(1-\lambda)}{\lambda/s+(1-\lambda)}
    \le 2\cdot \frac{\lambda/k+(1-\lambda)}{1-\lambda}
    = 2+\frac{2\lambda}{(1-\lambda)k}.
  \end{align*}
  Thus, by choosing $k=2$, the approximation ratio is at most $(2-\lambda)/(1-\lambda)$.
  For any $\epsilon >0$, by choosing $k\geq \frac{2}{\epsilon}\cdot\frac{\lambda}{1-\lambda}$,
  the approximation ratio is at most $2+\epsilon$.
\end{proof}

\paragraph*{(iii) $\boldsymbol{f(x)=x^2/(\lambda x+(1-\lambda))\ (\lambda \in[0,1])}$.}
This size function is derived by density function $\lambda \frac{w(S)}{|S|} + (1-\lambda)\frac{w(S)}{|S|^2}$. 
The following corollary provides an approximation ratio of our algorithm, which is at most $4$. 
\begin{corollary}\label{cor:ex3}
  For $f$-DS with $f(x)=x^2/(\lambda x+(1-\lambda))~(\lambda \in [0,1))$, 
  our algorithm has an approximation ratio of $4/(1+\lambda)$. 
\end{corollary}
\begin{proof}
  Let $s=|S^*|$.
  By Theorem~\ref{thm:convex}, the approximation ratio is  
  \begin{align*}
  \min\left\{\frac{f(2)/2}{f(s)/s^2},~\frac{2f(n)/n}{f(s)-f(s-1)}\right\}
  &=\min\left\{\frac{2(\lambda s+(1-\lambda))}{1+\lambda},~\frac{\frac{2n}{\lambda n+(1-\lambda)}}{\frac{s^2}{\lambda s+(1-\lambda)}-\frac{(s-1)^2}{\lambda (s-1)+(1-\lambda)}}\right\}\\
  &\le\min\left\{\frac{2(\lambda s+(1-\lambda))}{1+\lambda},~\frac{\frac{2n}{\lambda n+(1-\lambda)}}{\frac{s}{\lambda s+(1-\lambda)}}\right\}\\
  &\le\frac{2}{1+\lambda}\left(\lambda\cdot\frac{(1+\lambda)n}{\lambda n+(1-\lambda)}+(1-\lambda)\right)\\
  &\le\frac{2}{1+\lambda}\left(\lambda\cdot\frac{1+\lambda}{\lambda}+(1-\lambda)\right)
  =\frac{4}{1+\lambda}, 
  \end{align*}
  where the second inequality follows from the fact that 
  the first term and the second term of the minimum function are, respectively, 
  monotonically non-decreasing and non-increasing for $s$, 
  and they have the same value at $s=\frac{(1+\lambda)n}{\lambda n+(1-\lambda)}$.
\end{proof}

\section{Concave Case}\label{sec:concave}
In this section, we investigate $f$-DS with concave function $f$. 
A function $f:\mathbb{Z}_{\geq 0}\rightarrow \mathbb{R}_{\geq 0}$ is said to be \textit{concave} 
if $f(x)-2f(x+1)+f(x+2)\le 0$ holds for any $x\in\mathbb{Z}_{\geq 0}$.
We remark that $f(x)/x$ is monotonically non-increasing for $x$ since we assume that $f(0)=0$. 
It should be emphasized that any optimal solution to $f$-DS with concave function $f$ has a size 
larger than or equal to that of any densest subgraph. 
This follows from inequality \eqref{eq:size} and the monotonicity of $f(x)/x$.

\subsection{Dense Frontier Points} 
Here we define the dense frontier points and prove some basic properties. 
We denote by $\mathcal{P}$ the set \(\{(|S|,w(S))\mid S\subseteq V\}\).
A point in $\mathcal{P}$ is called a \emph{dense frontier point} 
if it is a unique maximizer of $y-\lambda x$ over $(x,y)\in \mathcal{P}$ for some $\lambda>0$.
In other words, the extreme points of the upper convex hull of $\mathcal{P}$ are dense frontier points.
The (smallest) densest subgraph is a typical subset of vertices corresponding to a dense frontier point.
We prove that 
(i) for any dense frontier point, there exists some concave function $f$ 
such that any optimal solution to $f$-DS with the function $f$ corresponds to the dense frontier point, 
and conversely, 
(ii) for any strictly concave function $f$ 
(i.e., $f$ that satisfies $f(x)-2f(x+1)+f(x+2)<0$ for any $x\in\mathbb{Z}_{\ge 0}$), 
any optimal solution to $f$-DS with the function $f$ corresponds to a dense frontier point. 

\begin{figure}[t]
\centering
\includegraphics[scale=1.0]{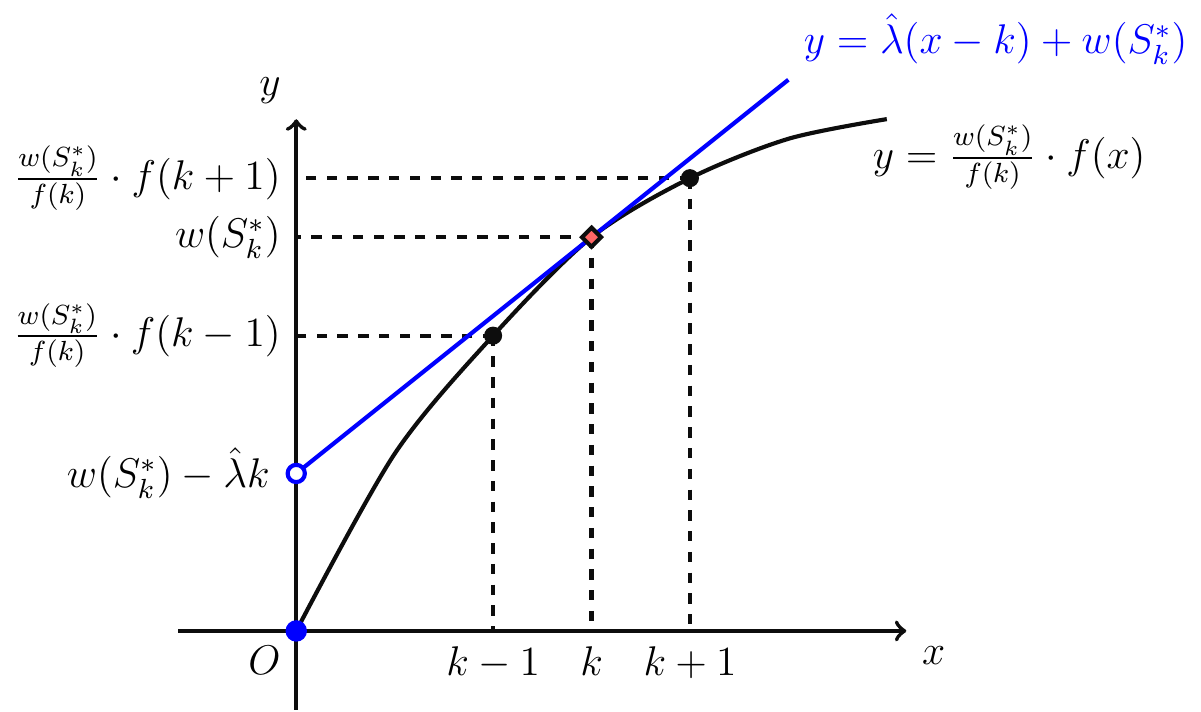}
\caption{A relationship between a dense frontier point and concave funcitons.}\label{fig:dfp}
\end{figure}

We first prove (i). 
Note that each dense frontier point can be written as $(i,w(S_i^*))$ for some $i\in \{0,1,\dots,n\}$, 
where $S_i^*\subseteq V$ is a maximum weight subset of size $i$.
Let $(k,w(S_k^*))$ be a dense frontier point 
and assume that it is a unique maximizer of $y-\hat{\lambda} x$ over $(x,y)\in \mathcal{P}$ for $\hat{\lambda}>0$.
Consider the concave function $f$ such that \(f(x)=\hat{\lambda}(x-k)+w(S_k^*)\) for $x>0$ and $f(0)=0$ (see Figure~\ref{fig:dfp}).
The concavity of $f$ follows from $w(S_k^*)-\hat{\lambda}k\ge w(S_0^*)-\hat{\lambda}\cdot 0=0=f(0)$. 
Then, any optimal solution $S^*\subseteq V$ to $f$-DS with the function $f$ corresponds to 
the dense frontier point (i.e., \((|S^*|,w(S^*))=(k,w(S_k^*))\) holds) 
because \(w(S)/f(|S|)\) is greater than or equal to $1$ 
if and only if \(w(S)-\hat{\lambda}|S|\ge w(S_k^*)-\hat{\lambda}k\) holds.

We next prove (ii). 
Let $f$ be any strictly concave function. 
Let $S_k^*\subseteq V$ be any optimal solution to $f$-DS with the function $f$,  
and take $\hat{\lambda}$ that satisfies 
\((f(k)-f(k-1))\cdot\frac{w(S_k^*)}{f(k)}>\hat{\lambda}> (f(k+1)-f(k))\cdot\frac{w(S_k^*)}{f(k)}\) (see Figure~\ref{fig:dfp}). 
Note that the strict concavity of $f$ guarantees the existence of such $\hat{\lambda}$. 
Since $f$ is strictly concave, we have
\begin{align*}
  \hat{\lambda}(|S|-k)+w(S_k^*)
  \ge \frac{w(S_k^*)}{f(k)}\cdot f(|S|)\ge \frac{w(S)}{f(|S|)}\cdot f(|S|)=w(S)
\end{align*}
for any \(S\subseteq V\), and the inequalities hold as equalities only when $(|S|,w(S))=(k,w(S_k^*))$.
Thus, $(k,w(S_k^*))$ is a unique maximizer of $y-\hat{\lambda} x$ over \((x,y)\in \mathcal{P}\), 
and hence is a dense frontier point.

\subsection{LP-Based Algorithm}
We provide an LP-based polynomial-time exact algorithm. 
We introduce a variable $x_e$ for each $e\in E$ and a variable $y_v$ for each $v\in V$. 
For $k=1,\dots,n$, we construct the following linear programming problem: 
\begin{alignat*}{5}
&\LP_k:&\ \ &\text{maximize}  &\ &\sum_{e\in E}w(e)\cdot x_e &\quad &\\ 
&      &    &\text{subject to}&  &\sum_{v\in V}y_v=k, &&\\
&      &    &                 &  &x_e\le y_u,~x_e\le y_v  &&\text{for all }\ e=\{u,v\}\in E,\\
&      &    &                 &  &x_e,\,y_v\in [0,1] &&\text{for all }\ e\in E,\ v\in V.
\end{alignat*}
For an optimal solution $(\bm{x}^k, \bm{y}^k)$ to $\LPk$ and a real parameter $r$, 
we define a sequence of subsets $S^k(r)=\{v\in V\mid y^k_v\geq r\}$. 
For $k=1,\dots,n$, our algorithm first solves $\LPk$ to obtain an optimal solution $(\bm{x}^k,\bm{y}^k)$, 
and then computes $r^*_k$ that maximizes $w(S^k(r))/f(|S^k(r)|)$. 
Note here that to find such $r^*_k$, 
it suffices to check all the distinct sets $S^k(r)$ by simply setting $r=y^k_v$ for every $v\in V$. 
The algorithm returns $S\in \{S^1(r^*_1),\dots,S^n(r^*_n)\}$ that maximizes $w(S)/f(|S|)$. 
For reference, we describe the procedure in Algorithm \ref{alg:concave_lp}. 
Clearly, the algorithm runs in polynomial time. 
\begin{algorithm}[t] 
\caption{LP-based algorithm}\label{alg:concave_lp}
\For{$k\ot 1,\dots, n$}{
Solve $\LP_k$ and obtain an optimal solution $(\bm{x}^k,\bm{y}^k)$\;
Compute $r^*_k$ that maximizes $w(S^k(r))/f(|S^k(r)|)$\;
}
\Return $S\in \{S^1(r^*_1),\dots, S^n(r^*_n)\}$ that maximizes $w(S)/f(|S|)$\;
\end{algorithm}

In what follows, we demonstrate that 
Algorithm~\ref{alg:concave_lp} obtains an optimal solution to $f$-DS with concave function $f$. 
The following lemma provides a lower bound on the optimal value of $\LPk$. 
\begin{lemma}\label{lem:relax_concave}
For any $S\subseteq V$,
the optimal value of $\LP_{|S|}$ is at least $w(S)$.
\end{lemma}
\begin{proof}
	For $S\subseteq V$, we construct a solution $(\bm{x},\bm{y})$ of $\LP_{|S|}$ as follows: 
  \begin{align*}
    x_e=\begin{cases}
		1&\text{if } e\in E(S),\\
    0&\text{otherwise},
    \end{cases}
    \quad\text{and}\quad
    y_v=\begin{cases}
	1&\text{if } v\in S,\\
    0&\text{otherwise}.
    \end{cases}    
  \end{align*}
  Then we can easily check that $(\bm{x},\bm{y})$ is feasible for $\LP_{|S|}$ and its objective value is $w(S)$.
  Thus, we have the lemma.
\end{proof}

We prove the following key lemma. 
\begin{lemma}\label{lemma:concave_lp}
Let $S^*\subseteq V$ be an optimal solution to $f$-DS with concave function $f$, and let $k^*=|S^*|$. 
Furthermore, let $(\bm{x}^*,\bm{y}^*)$ be an optimal solution to $\LP_{k^*}$.
Then, there exists a real number $r$ such that $S^{k^*}(r)$ is optimal to $f$-DS with concave function $f$.
\end{lemma}
\begin{proof}
For each $e=\{u,v\}\in E$, we have $x^*_e=\min\{y^*_u,y^*_v\}$ from the optimality of $(x^*,y^*)$. 
Without loss of generality, we relabel the indices of $(\bm{x}^*, \bm{y}^*)$ so that $y^*_1\ge \dots\ge y^*_n$.
Then we have  
\begin{align}\label{ineq:numerator}
\int_0^{y_1^*} w(S^{k^*}(r)) dr
&= \int_0^{y_1^*} \left(\sum_{e=\{u,v\}\in E} w(e)\cdot [y^*_u\ge r~\text{and}~y^*_v\ge r]\right) dr \nonumber\\
&= \sum_{e=\{u,v\}\in E} \int_0^{y_1^*} \left(w(e)\cdot [y^*_u\ge r~\text{and}~y^*_v\ge r]\right)dr \nonumber\\
&= \sum_{e=\{u,v\}\in E} w(e)\cdot \min\{y^*_u,y^*_v\}= \sum_{e\in E} w(e)\cdot x_e^*\ge w(S^*), 
\end{align}
where $[y^*_u\ge r~\text{and}~y^*_v\ge r]$ is the function of $r$ 
that takes 1 if the condition in the square bracket is satisfied and 0 otherwise,
and the last inequality follows from Lemma~\ref{lem:relax_concave}.
Moreover, we have
\begin{align}\label{ineq:denominator}
\int_0^{y_1^*} f(|S^{k^*}(r)|) dr &=\sum_{h=1}^n f(h)\cdot (y_h^*-y_{h+1}^*)=\sum_{h=1}^n (f(h)-f(h-1))\cdot y_h^*\nonumber\\
&\le \sum_{h=1}^{k^*} (f(h)-f(h-1))=f(k^*)-f(0)=f(k^*), 
\end{align}
where $y_{n+1}^*$ is defined to be $0$ for convenience, 
and the inequality holds by the concavity of $f$ (i.e., $f(h+2)-f(h+1)\le f(h+1)-f(h)$),
$\sum_{h=1}^n y^*_h=k^*$, and $y^*_h\le 1$.

Let $r^*$ be a real number that maximizes \(w(S^{k^*}(r))/f(|S^{k^*}(r)|)\) in \([0,y_1^*]\). 
Using inequalities (\ref{ineq:numerator}) and (\ref{ineq:denominator}), we have 
\begin{align*}
\frac{w(S^*)}{f(k^*)}
&\le \frac{\int_0^{y_1^*} w(S^{k^*}(r)) dr}{\int_0^{y_1^*} f(|S^{k^*}(r)|) dr}
= \frac{\int_0^{y_1^*} \left(\frac{w(S^{k^*}(r))}{f(|S^{k^*}(r)|)}\cdot f(|S^{k^*}(r)|)\right) dr}{\int_0^{y_1^*} f(|S^{k^*}(r)|) dr}\\
&\le \frac{\int_0^{y_1^*} \left(\frac{w(S^{k^*}(r^*))}{f(|S^{k^*}(r^*)|)}\cdot f(|S^{k^*}(r)|)\right) dr}{\int_0^{y_1^*} f(|S^{k^*}(r)|) dr}
= \frac{w(S^{k^*}(r^*))}{f(|S^{k^*}(r^*)|)}.
\end{align*}
This completes the proof. 
\end{proof}

Algorithm~\ref{alg:concave_lp} considers $S^{k^*}(r^*)$ as a candidate subset of the output. 
Therefore, we have the desired result. 

\begin{theorem}
Algorithm \ref{alg:concave_lp} is a polynomial-time exact algorithm
for $f$-DS with concave function $f$.
\end{theorem}

By Lemma \ref{lemma:concave_lp}, 
for any concave function $f$, an optimal solution to $f$-DS with the function $f$ 
is contained in \(\{S^k(r)\mid k=1,\dots,n,~r\in[0,1]\}\) whose cardinality is at most $n^2$. 
As shown above, for any dense frontier point, there exists some concave function $f$ such that 
any optimal solution to $f$-DS with the function $f$ corresponds to the dense frontier point. 
Thus, we have the following result. 
\begin{theorem}
We can find a corresponding subset of vertices for every dense frontier point in polynomial time. 
\end{theorem}

\subsection{Flow-Based Algorithm}
We provide a combinatorial exact algorithm for unweighted graphs (i.e., $w(e)=1$ for every $e\in E$). 
We first show that using the standard technique for fractional programming, 
we can reduce $f$-DS with concave function $f$ to a sequence of submodular function minimizations. 
The critical fact is that $\max_{S\subseteq V} w(S)/f(|S|)$ is at least $\beta$ 
if and only if $\min_{S\subseteq V} (\beta \cdot f(|S|)-w(S))$ is at most $0$.
Note that for $\beta \geq 0$, the function \(\beta \cdot f(|S|)-w(S)\) is submodular 
because \(\beta \cdot f(|S|)\) and \(-w(S)\) are submodular~\cite{fujishige_05}.
Thus, we can calculate \(\min_{S\subseteq V}(\beta \cdot f(|S|)-w(S))\) in $O(n^5(m+n))$ time 
using Orlin's algorithm \cite{Orlin_09}, 
which implies that we can determine \(\max_{S\subseteq V}w(S)/f(|S|)\ge \beta\) or not in $O(n^5(m+n))$ time.
Hence, we can obtain the value of \(\max_{S\subseteq V}w(S)/f(|S|)\) by binary search.
Note that the objective function of $f$-DS on unweighted graphs may have at most $O(mn)$ distinct values
since $w(S)$ is a nonnegative integer at most $m$.
Thus, the procedure yields an optimal solution in \(O(\log(mn))=O(\log n)\) iterations.
The total running time is $O(n^5(m+n)\cdot \log n)$.

To reduce the computation time, 
we replace Orlin's algorithm with a much faster flow-based algorithm 
that substantially extends a technique of Goldberg's flow-based algorithm 
for the densest subgraph problem~\cite{Goldberg_84}. 
The key technique is to represent the value of \(\min_{S\subseteq V}(\beta \cdot f(|S|)-w(S))\) 
using the cost of minimum cut of a certain directed network constructed from $G$ and $\beta \geq 0$. 

For a given unweighted undirected graph $G=(V,E,w)$ (i.e., $w(e)=1$ for every $e\in E$) and a real number $\beta \ge 0$, 
we construct a directed network $(U,A,w_{\beta})$ as follows. 
Note that for later convenience, we discuss the procedure on weighted graphs. 
The vertex set $U$ is defined by $U=V\cup P\cup\{s,t\}$, where $P=\{p_1,\dots,p_n\}$.
The edge set $A$ is given by $A=A_s\cup A_t\cup A_1\cup A_2$, where
\begin{align*}
A_s&=\{(s,v)\mid v\in V\},\ A_t=\{(p,t)\mid p\in P\},\\
	A_1&=\{(u,v), (v,u)\mid \{u,v\}\in E\},\ \text{and } A_2=\{(v,p)\mid v\in V,~p\in P\}.
\end{align*}
The edge weight $w_{\beta}:A\to\mathbb{R}_{\geq 0}$ is defined by
\begin{align*}
  w_{\beta}(e)=
  \begin{cases}
  d(v)/2 &(e=(s,v)\in A_s),\\
  \beta\cdot k\cdot a_k &(e=(p_k,t)\in A_t),\\
  1/2\ (= w(\{u,v\})/2)&(e=(u,v)\in A_1),\\ 
  \beta\cdot a_k &(e=(v,p_k)\in A_2),\\  
  \end{cases}
\end{align*}
where $d(v)$ is the (weighted) degree of vertex $v$, and 
\begin{align*}
  a_k=\begin{cases}
  2f(k)-f(k+1)-f(k-1) &(k=1,\dots,n-1),\\
  f(n)-f(n-1) & (k=n).
  \end{cases}
\end{align*}
Note that $a_k\geq 0$ holds since $f$ is a monotonically non-decreasing concave function. 
For reference, Figure~\ref{fig:mincut} depicts the network $(U,A,w_\beta)$. 
\begin{figure}[t]
\centering
\includegraphics[scale=1.05]{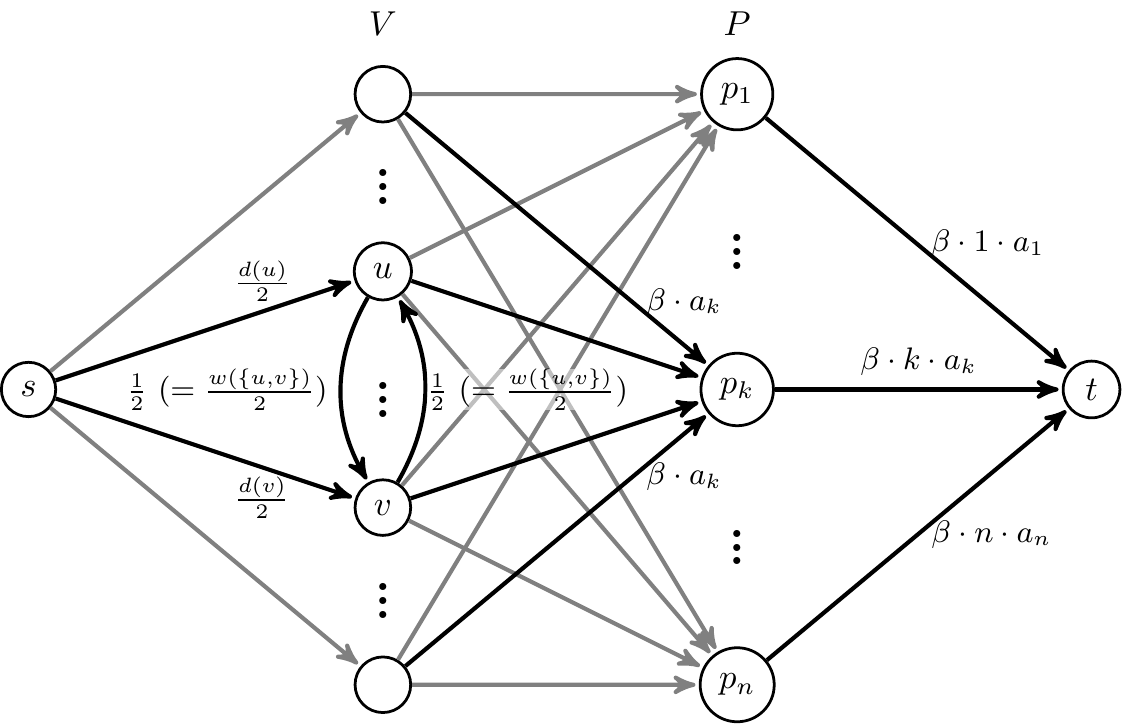}
\caption{The network $(U,A,w_\beta)$ constructed from $G$ and $\beta\geq 0$.} \label{fig:mincut}
\end{figure}

The following lemma reveals the relationship between 
a minimum $s$--$t$ cut in $(U,A,w_\beta)$ and the value of $\min_{S\subseteq V}(\beta \cdot f(|S|) - w(S))$. 
Note that an \emph{$s$--$t$ cut} in $(U,A,w_\beta)$ is a partition $(X,Y)$ of $U$ (i.e., $X\cup Y=U$ and $X\cap Y=\emptyset$) 
such that $s\in X$ and $t\in Y$, 
and the \emph{cost} of $(X,Y)$ is defined to be $\sum_{(u, v)\in A:\,u\in X, v\in Y}w_\beta(u,v)$. 
\begin{lemma}\label{lem:mincut}
  Let $(X,Y)$ be any minimum $s$--$t$ cut in the network $(U,A,w_\beta)$, and let $S=X\cap V$.  
  Then, the cost of $(X,Y)$ is equal to $w(V)+\beta \cdot f(|S|)-w(S)$. 
\end{lemma}

\begin{proof}
We first show that for any positive integer $s~(\le n)$, it holds that 
\begin{align}\label{eq:a_k}
	\sum_{i=1}^n \min\{i,s\}\cdot a_i=f(s). 
\end{align}
  By the definition of $a_k$, we get
  \begin{align*}
    \sum_{j=i}^n a_j
    &=(f(n)-f(n-1))+\sum_{j=i}^{n-1} (2f(j)-f(j+1)-f(j-1))\\
    &=(f(n)-f(n-1))-\sum_{j=i}^{n-1} ((f(j+1)-f(j))-(f(j)-f(j-1)))
	=f(i)-f(i-1).
  \end{align*}
  Thus, we have
  \begin{align*}
    \sum_{i=1}^n \min\{i,s\}\cdot a_i
    =\sum_{i=1}^s\sum_{j=i}^n a_j
	=\sum_{i=1}^s(f(i)-f(i-1))
	=f(s)-f(0)=f(s).
  \end{align*}

We are now ready to prove the lemma. 
Note that $p_k\in X$ if $|S|>k$ and $p_k\in Y$ if $|S|<k$.
Therefore, the cost of the minimum cut $(X,Y)$ is
\begin{align*}
&\sum_{v\in V\setminus S}\frac{d(v)}{2}
  +\sum_{\{u,v\}\in E:\,u\in S,\,v\in V\setminus S}\frac{w(\{u,v\})}{2}
  +\beta \cdot \sum_{i=1}^n \min\{i,|S|\}\cdot a_i\\
&=\sum_{v\in V\setminus S}\frac{d(v)}{2}
  +\sum_{\{u,v\}\in E:\,u\in S,\,v\in V\setminus S}\frac{w(\{u,v\})}{2}
  +\beta \cdot f(|S|)\\
&=\sum_{\{u,v\}\in E} w(\{u,v\})
  -\sum_{\{u,v\}\in E(S)}w(\{u,v\})
  +\beta \cdot f(|S|)
  =w(V)+\beta \cdot f(|S|)-w(S), 
\end{align*}
where the first equality follows from equality~(\ref{eq:a_k}).
\end{proof}

From this lemma, we see that the cost of a minimum $s$--$t$ cut is \(w(V)+\min_{S\subseteq V}(\beta \cdot f(|S|)-w(S))\). 
Therefore, for a given value $\beta \geq 0$, we can determine whether there exists $S\subseteq V$ 
that satisfies $w(S)/f(|S|)\geq \beta$ 
by checking the cost of a minimum $s$--$t$ cut is at most $w(V)$ or not. 
Our algorithm applies binary search for $\beta$ within the possible objective values of $f$-DS 
(i.e., $\{p/f(q)\mid p=0,1,\dots,m,~q=2,3,\dots,n\}$). 
For reference, we describe the procedure in Algorithm~\ref{alg:concave_flow}. 
The minimum $s$--$t$ cut problem can be solved in $O(N^3/\log N)$ time for a network with $N$ vertices \cite{Cheriyan+96}.
Thus, the running time of our algorithm is $O(\frac{n^3}{\log n}\cdot\log(mn))=O(n^3)$ since $|U|=2n+2$.
We summarize the result in the following theorem. 
\begin{algorithm}[t] 
\caption{Flow-based algorithm for unweighted graphs}\label{alg:concave_flow}
Let $\{\beta_1,\dots,\beta_r\}=\{p/f(q)\mid p=0,1,\dots,m,~q=2,3,\dots,n\}$ such that \(\beta_1<\dots< \beta_r\)\;
$i_{\min}\ot 1$ and $i_{\max}\ot r$\;
\While{\textsc{True}}{
  $i\ot \lfloor(i_{\max}+i_{\min})/2\rfloor$\;
  Compute a minimum $s$--$t$ cut $(X,Y)$ in $(U,A,w_{\beta_i})$\;
  \lIf{the cost of $(X,Y)$ is larger than $w(V)$}{$i_{\max}\ot i-1$}
  \lElseIf{the cost of $(X,Y)$ is less than $w(V)$}{$i_{\min}\ot i+1$}
  \lElse{\Return $X\cap V$}
}
\end{algorithm}

\begin{theorem}
Algorithm~\ref{alg:concave_flow} is an $O(n^3)$-time exact algorithm for $f$-DS with concave function $f$ on unweighted graphs.
\end{theorem}

For $f$-DS with concave function $f$ on weighted graphs, 
the binary search used in Algorithm~\ref{alg:concave_flow} is not applicable 
because there may be exponentially many possible objective values in the weighted setting. 
Alternatively, we present an algorithm that employs another binary search strategy (Algorithm~\ref{alg:concave_flow_weighted}). 
We have the following theorem.

\newcommand{\low}{\beta_{\mathrm{lb}}}
\newcommand{\hi}{\beta_{\mathrm{ub}}}
\begin{algorithm}[t] 
\caption{Flow-based algorithm for weighted graphs}\label{alg:concave_flow_weighted}
$\low^{(0)}\ot w(S_2^*)/f(2)$, $\hi^{(0)}\ot w(V)/f(2)$, and $i\ot 0$\;
\While{$\hi^{(i)}\ge (1+\epsilon)\cdot \low^{(i)}$}{
  $\beta^{(i)}\ot \sqrt{\low^{(i)}\cdot\hi^{(i)}}$\;
  Compute a minimum $s$--$t$ cut $(X^{(i)},Y^{(i)})$ in $(U,A,w_{\beta^{(i)}})$\;
  \lIf{the cost of $(X^{(i)},Y^{(i)})$ is larger than $w(V)$}{$\low^{(i+1)}\ot \low^{(i)}$ and $\hi^{(i+1)}\ot \beta^{(i)}$}
  \lElse{$\low^{(i+1)}\ot \beta^{(i)}$ and $\hi^{(i+1)}\ot \hi^{(i)}$}
  $i\ot i+1$\;
}
Compute a minimum $s$--$t$ cut $(X^{(i)},Y^{(i)})$ in $(U,A,w_{\low^{(i)}})$ and \textbf{return} $X^{(i)}\cap V$\;
\end{algorithm}

\begin{theorem}
Algorithm~\ref{alg:concave_flow_weighted} is an $O\left(\frac{n^3}{\log n}\cdot \log\left(\frac{\log n}{\epsilon}\right)\right)$-time $(1+\epsilon)$-approximation algorithm for $f$-DS with concave function $f$.
\end{theorem}
\begin{proof}
Let $i^*$ be the number of iterations executed by Algorithm~\ref{alg:concave_flow_weighted}, and let $\hat{S}=X^{(i^*)}\cap V$.
Then we have $\max_{S\subseteq V} w(S)/f(|S|)\le \hi^{(i^*)}$, $\hi^{(i^*)}<(1+\epsilon)\cdot \low^{(i^*)}$, and $\low^{(i^*)}\le w(\hat{S})/f(|\hat{S}|)$.
Combining these inequalities, we have  
$\max_{S\subseteq V} w(S)/f(|S|)< (1+\epsilon)\cdot w(\hat{S})/f(|\hat{S}|)$, 
which means that $\hat{S}$ is a $(1+\epsilon)$-approximate solution.

In what follows, we analyze the time complexity of the algorithm. For each $i\in\{0,1,\dots,i^*-1\}$, it holds that 
\begin{align*}
	\frac{\hi^{(i+1)}}{\low^{(i+1)}}
	\leq \max\left\{\frac{\hi^{(i)}}{\beta^{(i)}},~\frac{\beta^{(i)}}{\low^{(i)}}\right\}
	=\max\left\{\frac{\hi^{(i)}}{\sqrt{\low^{(i)}\cdot\hi^{(i)}}},~\frac{\sqrt{\low^{(i)}\cdot\hi^{(i)}}}{\low^{(i)}}\right\}
	= \sqrt{\frac{\hi^{(i)}}{\low^{(i)}}}. 
\end{align*}
Since $\hi^{(0)}/\low^{(0)}=w(V)/w(S_2^*)\leq m$, we have $\hi^{(i)}/\low^{(i)}\leq m^{1/2^i}$ for $i=1,\dots,i^*$. 
Note that $i^*$ is the minimum index $i$ that satisfies $\hi^{(i)}/\low^{(i)}\leq 1+\epsilon$. 
Thus, we see that $i^*$ is upper bounded by $O\left(\log\left(\frac{\log n}{\log(1+\epsilon)}\right)\right)$. 
Therefore, the total running time of the algorithm is 
$O\left(\frac{n^3}{\log n}\cdot \log\left(\frac{\log n}{\log(1+\epsilon)}\right)\right)
=O\left(\frac{n^3}{\log n}\cdot \log\left(\frac{\log n}{\epsilon}\right)\right)$, 
where the equality follows from the fact that $\lim_{\epsilon\to +0}\frac{\epsilon}{\log(1+\epsilon)}=1$ holds.
\end{proof}

\subsection{Greedy Peeling}
Finally, we provide an approximation algorithm with much higher scalability. 
Specifically, we prove that the greedy peeling (Algorithm~\ref{alg:peeling}) has an approximation ratio of 3 
for $f$-DS with concave function $f$. 
As mentioned above, the algorithm runs in $O(m+n\log n)$ time for weighted graphs and $O(m+n)$ time for unweighted graphs.

We prove the approximation ratio. 
Recall that $S_n,\dots, S_1$ are the subsets of vertices produced by the greedy peeling.
We use the following fact, 
which implies that there exists a $3$-approximate solution for Dal$k$S in $S_n,\dots, S_{k}$. 
\begin{fact}[Theorem~1 in Andersen and Chellapilla \cite{Andersen_Chellapilla_09}]\label{lem:ac}
  For any integer $k~(\le n)$, it holds that 
  \begin{align*}
	  \max_{S\subseteq V:\,|S|\ge k}\frac{w(S)}{|S|}\le 3\cdot \max_{k\le i\le n}\frac{w(S_i)}{i}.
  \end{align*}
\end{fact}

\begin{theorem}\label{thm:concave_peeling}
  The greedy peeling (Algorithm~\ref{alg:peeling}) has an approximation ratio of $3$ 
  for $f$-DS with concave function $f$.
\end{theorem}

\begin{proof}
Let $S^*\subseteq V$ be an optimal solution to $f$-DS with concave function $f$, and let $s^*=|S^*|$.
Let $S\subseteq V$ be the output of the greedy peeling for the problem. 
Then we have
\begin{align*}
  \frac{w(S^*)}{f(s^*)}
  &=\frac{w(S^*)/s^*}{f(s^*)/s^*}
  \leq \max_{S'\subseteq V:\,|S'|\ge s^*}\frac{w(S')/|S'|}{f(s^*)/s^*}\\
  &\le 3\cdot \max_{s^*\le i\le n} \frac{w(S_i)/i}{f(s^*)/s^*}
  \le 3\cdot \max_{s^*\le i\le n}\frac{w(S_i)/i}{f(i)/i} 
  \le 3\cdot\frac{w(S)}{f(|S|)}, 
\end{align*}
where the second inequality follows from Fact~\ref{lem:ac}, and the third inequality follows from the monotonicity of $f(x)/x$. 
\end{proof}

\section*{Acknowledgments}
The authors would like to thank Yoshio Okamoto for pointing out the reference~\cite{NaKaAi11}. 
The first author is supported by a Grant-in-Aid for Young Scientists (B) (No.~16K16005). 
The second author is supported by a Grant-in-Aid for JSPS Fellows (No.~26-11908).

\bibliographystyle{abbrv}
\bibliography{kawase_miyauchi}

\begin{thebibliography}{10}

\bibitem{Andersen_Chellapilla_09}
R.~Andersen and K.~Chellapilla.
\newblock Finding dense subgraphs with size bounds.
\newblock In {\em WAW~'09: Proceedings of the 6th Workshop on Algorithms and
  Models for the Web Graph}, pages 25--37, 2009.

\bibitem{Angel+_12}
A.~Angel, N.~Sarkas, N.~Koudas, and D.~Srivastava.
\newblock Dense subgraph maintenance under streaming edge weight updates for
  real-time story identification.
\newblock In {\em VLDB~'12: Proceedings of the 38th International Conference on
  Very Large Data Bases}, pages 574--585, 2012.

\bibitem{Asahiro+_00}
Y.~Asahiro, K.~Iwama, H.~Tamaki, and T.~Tokuyama.
\newblock Greedily finding a dense subgraph.
\newblock {\em J. Algorithms}, 34(2):203--221, 2000.

\bibitem{Bader_Hogue_03}
G.~D. Bader and C.~W.~V. Hogue.
\newblock An automated method for finding molecular complexes in large protein
  interaction networks.
\newblock {\em BMC Bioinformatics}, 4(1):1--27, 2003.

\bibitem{Bhaskara+_10}
A.~Bhaskara, M.~Charikar, E.~Chlamtac, U.~Feige, and A.~Vijayaraghavan.
\newblock Detecting high log-densities: An ${O}(n^{1/4})$ approximation for
  densest $k$-subgraph.
\newblock In {\em STOC~'10: Proceedings of the 42nd ACM Symposium on Theory of
  Computing}, pages 201--210, 2010.

\bibitem{Bonchi+_14}
F.~Bonchi, F.~Gullo, A.~Kaltenbrunner, and Y.~Volkovich.
\newblock Core decomposition of uncertain graphs.
\newblock In {\em KDD~'14: Proceedings of the 20th ACM SIGKDD International
  Conference on Knowledge Discovery and Data Mining}, pages 1316--1325, 2014.

\bibitem{Charikar_00}
M.~Charikar.
\newblock Greedy approximation algorithms for finding dense components in a
  graph.
\newblock In {\em APPROX~'00: Proceedings of the 3rd International Workshop on
  Approximation Algorithms for Combinatorial Optimization}, pages 84--95, 2000.

\bibitem{Cheriyan+96}
J.~Cheriyan, T.~Hagerup, and K.~Mehlhorn.
\newblock An $o(n^3)$-time maximum-flow algorithm.
\newblock {\em SIAM J. Comput.}, 25(6):1144--1170, 1996.

\bibitem{Dourisboure+_07}
Y.~Dourisboure, F.~Geraci, and M.~Pellegrini.
\newblock Extraction and classification of dense communities in the web.
\newblock In {\em WWW~'07: Proceedings of the 16th International Conference on
  World Wide Web}, pages 461--470, 2007.

\bibitem{FeKoPe01}
U.~Feige, D.~Peleg, and G.~Kortsarz.
\newblock The dense $k$-subgraph problem.
\newblock {\em Algorithmica}, 29(3):410--421, 2001.

\bibitem{Fratkin+_06}
E.~Fratkin, B.~T. Naughton, D.~L. Brutlag, and S.~Batzoglou.
\newblock Motif{C}ut: regulatory motifs finding with maximum density subgraphs.
\newblock {\em Bioinformatics}, 22(14):e150--e157, 2006.

\bibitem{fujishige_05}
S.~Fujishige.
\newblock {\em Submodular {F}unctions and {O}ptimization}, volume~58 of {\em
  Annals of Discrete Mathematics}.
\newblock Elsevier, 2005.

\bibitem{Gibson+_05}
D.~Gibson, R.~Kumar, and A.~Tomkins.
\newblock Discovering large dense subgraphs in massive graphs.
\newblock In {\em VLDB~'05: Proceedings of the 31st International Conference on
  Very Large Data Bases}, pages 721--732, 2005.

\bibitem{Goldberg_84}
A.~V. Goldberg.
\newblock Finding a maximum density subgraph.
\newblock Technical report, University of California Berkeley, 1984.

\bibitem{Kawase_Miyauchi_16}
Y.~Kawase and A.~Miyauchi.
\newblock The densest subgraph problem with a convex/concave size function.
\newblock In {\em ISAAC~'16: Proceedings of the 27th International Symposium on
  Algorithms and Computation}, pages 44:1--44:12, 2016.

\bibitem{Khot_06}
S.~Khot.
\newblock Ruling out {PTAS} for graph min-bisection, dense $k$-subgraph, and
  bipartite clique.
\newblock {\em SIAM J. Comput.}, 36(4):1025--1071, 2006.

\bibitem{Khuller_Saha_09}
S.~Khuller and B.~Saha.
\newblock On finding dense subgraphs.
\newblock In {\em ICALP~'09: Proceedings of the 36th International Colloquium
  on Automata, Languages and Programming}, pages 597--608, 2009.

\bibitem{NaKaAi11}
K.~Nagano, Y.~Kawahara, and K.~Aihara.
\newblock Size-constrained submodular minimization through minimum norm base.
\newblock In {\em ICML~'11: Proceedings of the 28th International Conference on
  Machine Learning}, pages 977--984, 2011.

\bibitem{Orlin_09}
J.~B. Orlin.
\newblock A faster strongly polynomial time algorithm for submodular function
  minimization.
\newblock {\em Math. Program.}, 118(2):237--251, 2009.

\bibitem{Tsourakakis+_13}
C.~E. Tsourakakis, F.~Bonchi, A.~Gionis, F.~Gullo, and M.~Tsiarli.
\newblock Denser than the densest subgraph: Extracting optimal quasi-cliques
  with quality guarantees.
\newblock In {\em KDD~'13: Proceedings of the 19th ACM SIGKDD International
  Conference on Knowledge Discovery and Data Mining}, pages 104--112, 2013.

\bibitem{Yanagisawa_Hara_16}
H.~Yanagisawa and S.~Hara.
\newblock Axioms of density: How to define and detect the densest subgraph.
\newblock Technical report, IBM Research - Tokyo, 2016.

\end{thebibliography}

\end{document}